\documentclass[a4paper,10pt]{article}
\usepackage[utf8]{inputenc}
\usepackage{amsfonts}
\usepackage{epstopdf}
\usepackage{graphicx}
\usepackage{bm}
\usepackage{xcolor}
\usepackage{bbold}
\usepackage{amsmath}
\usepackage{mathtools}
\usepackage{float}
\usepackage{extarrows}
\usepackage{amsthm}
\usepackage{graphicx}
\usepackage{calligra,mathrsfs}
\usepackage{amssymb}
\usepackage{tikz-cd}
\usepackage{enumitem}
\usepackage{breqn}
\usepackage{xr}
\usepackage{authblk}

\newtheorem*{theorem*}{Theorem}
\newtheorem*{lem}{Lemma}
\newtheorem{definition}{Definition}
\newtheorem{remark}{Remark}
\newtheorem{theorem}{Theorem}
\newtheorem{lemma}{Lemma}
\newtheorem{corollary}{Corollary}

\newcommand{\diag}{\mathrm{diag}}
\newcommand{\R}{\mathbb{R}}
\newcommand{\dd}{\mathrm{d}}

\newcommand{\Img}{\mathrm{Im}}
\newcommand{\Ker}{\mathrm{Ker}}

\newcommand{\Rq}{\mathbb{R}^q}

\newcommand{\com}[1]{\textcolor{red}{#1}}
\newcommand{\eqnref}[1]{Eq. (\ref{#1})}

\newcommand{\Vss}{\mathcal{V}^{ss}}
\newcommand{\A}{A}

\newcommand{\UU}{U}
\newcommand{\uu}{u}
\newcommand{\VV}{U}
\newcommand{\G}{V}

\newcommand{\HH}{H}
\newcommand{\X}{\diag(x)}
\newcommand{\XX}{\diag \left( \frac{1}{x} \right)}
\newcommand{\aij}{\alpha_{i \rightarrow j} }
\newcommand{\aii}{\alpha_{i \rightarrow i} }
\newcommand{\aji}{\alpha_{j \rightarrow i} }
\newcommand{\ai}{\alpha_i }

\makeatletter
\newcounter{savesection}
\newcounter{apdxsection}
\renewcommand\appendix{\par
  \setcounter{savesection}{\value{section}}%
  \setcounter{section}{\value{apdxsection}}%
  \setcounter{subsection}{0}%
  \gdef\thesection{\@Alph\c@section}}
\newcommand\unappendix{\par
  \setcounter{apdxsection}{\value{section}}%
  \setcounter{section}{\value{savesection}}%
  \setcounter{subsection}{0}%
  \gdef\thesection{\@arabic\c@section}}
\makeatother

\title{Cramer-Rao bound and absolute sensitivity in chemical reaction networks}
\author[1]{Dimitri Loutchko\thanks{d.loutchko@edu.k.u-tokyo.ac.jp}}
\author[2]{Yuki Sughiyama}
\author[1,3]{Tetsuya J. Kobayashi\thanks{http://research.crmind.net}}
\affil[1]{{Institute of Industrial Science, The University of Tokyo, 4-6-1, Komaba, Meguro-ku, Tokyo 153-8505 Japan.}}
\affil[2]{Graduate School of Information Sciences, Tohoku University, Sendai 980-8579, Japan}
\affil[3]{Department of Mathematical Informatics,
Graduate School of Information Science and Technology,
The University of Tokyo, Tokyo 113-8654, Japan.}
\affil[4]{Universal Biology Institute, The University of Tokyo,
7-3-1, Hongo, Bunkyo-ku, 113-8654, Japan.}
\date{}

\begin{document}

\maketitle

\begin{abstract}
Chemical reaction networks (CRN) comprise an important class of models to understand biological functions such as cellular information processing, the robustness and control of metabolic pathways, circadian rhythms, and many more.
However, any CRN describing a certain function does not act in isolation but is a part of a much larger network and as such is constantly subject to external changes.
In [Shinar, Alon, and Feinberg. "Sensitivity and robustness in chemical reaction networks." SIAM J App Math (2009): 977-998.], the responses of CRN to changes in the linear conserved quantities, called sensitivities, were studied in and the question of how to construct absolute, i.e., basis-independent, sensitivities was raised.
In this article, by applying information geometric methods, such a construction is provided.
The idea is to track how concentration changes in a particular chemical propagate to changes of all concentrations within a steady state.
This is encoded in the matrix of absolute sensitivities.
As the main technical tool, a multivariate Cramer-Rao bound for CRN is proven, which is based on the the analogy between quasi-thermostatic steady states and the exponential family of probability distributions.
This leads to a linear algebraic characterization of the matrix of absolute sensitivities for quasi-thermostatic CRN.
As an example, the core module of the IDHKP-IDH glyoxylate bypass regulation system is analyzed analytically and numerically by extensive random sampling of the concentration space.
The experimentally known findings for the robustness of the IDH enzyme are confirmed and a hidden symmetry at the level of distributions is revealed, providing a blueprint for the analysis of the robustness properties of CRNs.
\end{abstract}

\section{Introduction}
The theory of chemical reaction networks (CRNs) is an indispensable tool for the understanding of biochemical phenomena.
Models based on CRN have successfully been employed to gain insights into various biological functions such as cell signaling pathways \cite{gross2016, perez2018}, circadian clocks \cite{gonze2006,hatakeyama2012}, proofreading kinetics \cite{hopfield1974,banerjee2017}, and many more \cite{mikhailov2017,alon2019}.

The thorough mathematical basis of chemical reaction theory has enabled much of its continuous success in applications.
Originating in the fundamental work of Horn, Jackson, and Feinberg \cite{horn1972,feinberg1972complex,horn1972necessary,feinberg2019foundations}, modern CRN theory uses the theory of differential equations \cite{feinberg1980chemical,mincheva2008multigraph}, algebraic geometry \cite{craciun2009toric,craciun2022disguised,craciun2023structure}, graph theory \cite{craciun2006multiple,craciun2011graph}, homological algebra \cite{hirono2021}, and, in the recent years, information geometry.
By realizing that the vector of chemical concentration is a distribution on a finite set, it was possible to adapt techniques from information theory and statistics to CRN theory \cite{yoshimura2021,sughiyama2022hessian,kobayashi2023information}.
In particular, the equilibrium and quasi-thermostatic steady state manifolds have the same mathematical shape as exponential families from information geometry \cite{amari2016,kobayashi2022kinetic}.
This approach has improved the geometrical understanding of growing systems \cite{sughiyama2022chemical}, thermodynamic uncertainty relations and speed limits \cite{loutchko2023geometry}, and finite time driving \cite{loutchko2022}.

To be more precise, the stoichiometric polytopes of a CRN should be thought of as analogues of a probability simplex:
A concentration vector has linear quantities which are conserved by the internal CRN dynamics, analogous to the conservation of probability.
However, the linear conserved quantities of a CRN can be controlled and externally changed.
From the mathematical point of view, this gives the concentration space two distinguished directions:
The direction of internal dynamics within stiochiometric polytopes and the complementary direction which shifts the polytopes.
The main technical contribution of this article is to construct a Cramer-Rao bound for quasi-thermostatic CRNs which is adapted to detect changes in the latter direction.

In CRN theory, the influence of an infinitesimal change of the linear conserved quantities on the steady state concentrations is known as sensitivity analysis \cite{shinar2009}.
The understanding of sensitivity is a prerequisite for the understanding of concentration robustness \cite{barkai1997robustness,alon1999robustness,anderson2014stochastic}, absolute concentration robustness \cite{shinar2010structural}, perfect adaptation \cite{hirono2023,khammash2021perfect}, and for chemical switches \cite{acar2008stochastic,tyson2008biological,reyes2022numerical}.
Sensitivity has been extensively analyzed in the literature \cite{shinar2007input,shinar2009, shinar2009robustness,shinar2010structural,mochizuki2015sensitivity,fiedler2015sensitivity,okada2016,hirono2021,joshi2022foundations,meshkat2022absolute, joshi2023reaction,craciun2023structure,hirono2023robust,kaihnsa2024absolute,puente2024absolute}, yet, the problem remains that the sensitivity matrix is not an object that is intrinsic to the CRN but depends on the choice of a basis of $\Ker[S^T]$ (where $S$ is the stoichiometric matrix of the CRN).
This means that the respective numerical values of the sensitivity matrix elements have no immediate physical meaning as, for example, arbitrary scaling of the basis vectors of $\Ker[S^T]$ will lead to arbitrary scaling of the respective elements.
In this article, basis independent quantities, termed {\it absolute sensitivities}, are defined and analyzed.
The definition is given for any steady state manifold whenever locally there is a continuously differentiable parametrization of the manifold by the vector of conserved quantities.

The idea behind the definition of absolute sensitivity is as follows:
If, at a steady state, the concentration of a chemical $X_i$ is perturbed by an amount of $\delta x_i$, a new steady state is adopted by the CRN.
Thereby, the perturbation $\delta x_i$ distributes to concentration changes of all chemicals, prescribed by the coupling through the linear conserved quantities.
The absolute sensitivity $\aij$ of the chemical $X_j$ with respect to the chemical $X_i$ quantifies the fraction of infinitesimal concentration change of $X_j$ after this redistribution, i.e., the concentration change of the chemical $X_j$ is given by $\aij \delta x_i$, to first order in $\delta x_i$.
In particular, the diagonal terms $\alpha_i := \aii$ quantify the fraction of concentration that remains with the chemicals $X_i$ in the shifted steady sate.
The independence of the choice of basis of $\Ker[S^T]$ for the absolute sensitivities is not only aesthetically pleasing but has immediate consequences for the analysis of concentration robustness, negative feedback, and (hyper)sensitivity:
In this article, it is proven that $\ai \in [0,1]$ holds for quasi-thermostatic CRN.
However, for systems with non-ideal thermodynamics, $\ai <0$ and $\ai >1$ can occur, corresponding to self-regulation and hypersensitivity, respectively.
Moreover, absolute concentration robustness poses very strict conditions on the CRN model, for example it is not compatible with the complex balancing condition \cite{shinar2010structural}
By employing the concept of absolute sensitivity, it is possible to find weaker and therefore more practically feasible notions of concentration robustness.
It might often be of biological interest that the concentration of a certain chemical $X_i$ is insensitive to concentration changes of a particular chemical $X_j$, which is captured by the vanishing of $\aji$.
In contrast, absolute concentration robustness in $X_i$ requires the $\aji$ to vanish for all $i$.

A second remarkable feature of absolute sensitivities, in addition to the independence of the $\Ker[S^T]$ basis, is that they do not explicitly depend on kinetic parameters\footnote{But they do so implicitly because the steady state depends on the kinetic parameters determine.}.
This leads to efficient procedures for sampling the space of physiologically feasible concentration vectors and to obtain resulting distributions of absolute sensitivities which allow to infer general robustness properties of each of the chemicals.
The linear algebraic formalism developed in this article allows for computationally efficient implementations as is illustrated by using the core module of the IDHKP-IDH glyoxylate bypass regulation system as an example.
Moreover, the approach reveals a surprising symmetry at the level of distributions which is not presents at the level of individual concentration vectors.

It is worth mentioning that quantities related to absolute sensitivities have been defined in the past:
In \cite{feliu2019}, Feliu has introduced sign-sensitivities for CRNs to analyze whether the steady state concentrations of $X_i$ increases or decreases when the concentration of $X_j$ is perturbed.
In other words, Feliu has analyzed the signs of the $\aji$.
However, the work was done from the point of view of computational algebraic geometry and therefore the information geometrical interpretation of the $\aji$ as well as the exact formulae in the case of complex balanced systems, as derived in this paper, were not accessible in the previous work.

The article is structured as follows:
This introductory section contains a paragraph introducing the mathematical notation and then provides a summary of the main results.
The terminology for CRN is introduced in Section \ref{sec:CRN}.
The absolute sensitivity is defined in Section \ref{sec:abs_sens} and its basic properties, including the basis independence, are proven.
In general, a CRN does not allow for an explicit parametrization of the steady state manifold by the conserved quantities and thus no explicit expressions for the absolute sensitivities exist.
However, quasi-thermostatic steady state manifolds can be analyzed with the information geometric techniques which lead to the formulation of a Cramer-Rao bound and to a linear algebraic characterization of the absolute sensitivities as scalar products between the projections of the canonical vectors to a certain linear space.
This is presented in Section \ref{sec:CRB}.
In Section \ref{sec:example}, an example is given to illustrate the concepts developed in this article.

\subsection*{Mathematical notation}
\paragraph{Shorthand notation}
For a vector $x = (x_1,\dotsc,x_n)^T \in \R^n$, functions are defined componentwise, i.e., $\exp x = (\exp x_1, \dotsc, \exp x_n)^T \in \R^n$.
The Hadamard product $x \circ y$ between $x, y \in \R^n$ is defined by componentwise multiplication as $x \circ y = (x_1y_1, \dotsc, x_ny_n)^T \in \R^n$.
\paragraph{Linear maps}A $n \times m$ matrix $L$ is identified with a linear map $L : \R^m \rightarrow \R^n$ as $e_i \mapsto \sum_{j=1}^n L_{ij}e'_j$, where $e_i$ and $e'_j$ are the respective canonical basis vectors of $\R^m$ and $\R^n$.
The transpose matrix $L^T$ defines a map $L^T: \R^n \rightarrow \R^m$ via $e'_j \mapsto \sum_{i=1}^m L_{ij} e_i$.
\paragraph{Bilinear products}
On the vector space $\R^n$, the standard bilinear product $\langle.,.\rangle: \R^n \times \R^n \rightarrow \R$ is given by $\langle x,y\rangle = x^T y = \sum_{i=1}^n x_i y_i$ for $x,y \in \R^n$.
For any positive semidefinite $n \times n$ matrix $g$, the bilinear product $\langle.,.\rangle_g: \R^n \times \R^n \rightarrow \R$ is defined as $\langle x,y\rangle_g := x^Tgy = \sum_{i,j=1}^n x_i g_{ij} y_j$.
\paragraph{The Jacobian}
For a map $f: M \rightarrow N$ between manifolds with (local) coordinate systems denoted by $(x_i)$ and $(y_j)$ respectively, the Jacobian at $m \in M$ is the linear map $D_m f: T_mM \rightarrow T_{f(m)} N$ between tangent spaces induced by $f$.
When both the map $f$ and the point $m$ are clear from the context, the Jacobian is also written as $$D_mf = \frac{\partial y}{\partial x}.$$

\subsection*{Summary of the main results}

Consider a CRN with chemicals $X_1,\dotsc,X_n$, let $S$ denote the stiochiometric matrix, and $X := \R^n_{>0}$ the concentration space with concentration vectors $(x_1,\dotsc,x_n)^T \in X$.
The steady state manifold $\Vss$ is the zero locus for the deterministic CRN dynamics $\frac{\dd x}{\dd t}= Sj$ with flux vector $j$.
Let $U$ be a $n \times q$ matrix of a complete set of basis vectors $u_1,\dotsc,u_q$ for $\Ker[S^T]$ which yields the vector of linear conserved quantities corresponding to the point $x \in X$ as $\eta = (\eta_1,\dotsc,\eta_q)^T = U^T x$.
It is assumed that there is a map $\beta$ which parameterizes the steady state manifold by the vector of conserved quantities as $x = \beta(\eta) \in \Vss$ (in other words, $\beta$ is a section of the map $U^T : X \rightarrow \Rq$ with $\Img[\beta] \subset \Vss$).
The sensitivity at a point $x \in \Vss$ is given by the matrix $\chi$ with elements $\chi_{ij} = \frac{\partial x_i}{\partial \eta_j}$, which is the Jacobian of the map $\beta$.
These basic concepts are introduced in Section \ref{sec:CRN}.

In Section \ref{sec:abs_sens}, the {\it absolute sensitivity} $\aij$ of a chemical $X_j$ with respect to $X_i$ is defined at a point $x = \beta(\eta) \in \Vss$.
It quantifies the concentration change of $X_j$ in the steady state caused by a concentration change $\delta x_i$ of $X_i$, to first order in $\delta x_i$.
This concentration leads to a change in the vector of conserved quantities $\Delta \eta$ given by $\Delta \eta_j = \sum_{i=1}^n u_{ij} \delta x_i$ and to the adjusted steady state
\begin{equation*}
    \beta(\eta + \Delta \eta) = \beta(\eta) + D_{\eta} \beta (\Delta \eta) + \mathcal{O}(\|\Delta \eta \|^2) = x + D_{\eta} \beta (U^T e_i) \delta x_i + \mathcal{O}(\delta x_i^2).
\end{equation*}
The absolute sensitivity $\aij$ is the $j$th component of the linear term $D_{\eta} \beta (U^T e_i)$ in Definition \ref{def:alpha_i}, which is explicitly given by
\begin{align*}
    \aij:= \sum_{k=1}^q \frac{\partial x_j}{\partial \eta_k} u_{ik}.
\end{align*}
The $n \times n$ matrix $\A$ of absolute sensitivities is given by $\A_{ij} = \aji$.
A diagonal element $\alpha_i := \aii$ and is called the {\it absolute sensitivity of the chemical $X_i$}.
The advantage of the absolute sensitivity over the classical sensitivity $\chi$ is stated in Theorem \ref{thm:abs_sens1}, which reads:
\begin{theorem*}
The matrix of absolute sensitivities $\A$ is independent of the choice of a basis of $\Ker[S^T]$.
Moreover, the equality
\begin{align*}
   \mathrm{Tr}[\A] = \sum_{i=1}^n \alpha_i = q
\end{align*}
holds, whereby $q = \dim \Ker[S^T]$.
\end{theorem*}

Section \ref{sec:CRB} treats the case of quasi-thermostatic CRN with methods from information theory.
Quasi-thermostatic CRN are characterized by the particular form
\begin{equation*}
    \Vss = \{ x \in X | \log x - \log x^{ss} \in \Ker[S^T] \}
\end{equation*}
of the steady state manifold and they allow a parametrization akin to the exponential family of probability distributions.
This analogy to information theory leads to a multivariate Cramer-Rao bound
\begin{equation*}
    \mathrm{Cov}(I_n) \geq \XX \A,
\end{equation*}
which is understood in the sense that the difference between the two matrices $\mathrm{Cov}(I_n)$ and $\XX \A$ is positive semidefinite.
This is stated and proven in Section \ref{sec:CRB_derivation}, Theorem \ref{thm:CRB}\footnote{The theorem states a more general version for the covariance of an arbitrary matrix $V$ instead of the identity matrix $I_n$ but this is not discussed in this introductory section.}.
Here, the covariance matrix elements are defined as $\mathrm{Cov}(I_n)_{ij}$ $ = \langle e_i - \bar{e}_i, e_j - \bar{e}_j \rangle_{\frac{1}{x}}$, where the bilinear form $\langle ., . \rangle_{\frac{1}{x}}$ is given by $\langle v,w \rangle_{\frac{1}{x}} := \sum_{i=1}^n  \frac{1}{x_i} v_i w_i.$, the $e_i$ are the canonical unit vectors, and the $\bar{e}_i$ are arbitrary vectors in $\Img[S]$.

In Section \ref{sec:tightening}, the Cramer-Rao bound is tightened by tuning the $\bar{e}_i$.
This leads to the linear-algebraic characterization of the absolute sensitivities $\alpha_i$ in Lemma \ref{lem:abs_sens2}:
\begin{lem}
For a quasi-thermostatic CRN, the absolute sensitivity $\alpha_i$ of the chemical $X_i$ at a point $x =(x_1,\dotsc,x_n)$ is given by
\begin{align*}
    \alpha_i = x_i \lVert \pi(e_i) \rVert_{\frac{1}{x}}^2,
\end{align*}
where $e_i$ is the $i$th canonical unit vector.
\end{lem}
Based on this lemma, the final Theorem \ref{thm:matrix_abs_sens} on the matrix of absolute sensitivities is proven:
\begin{theorem*}
    For a quasi-thermostatic CRN, the matrix of absolute sensitivities $\A$ at a point $x \in \Vss$ is given by
\begin{equation*}
    \A = \X \mathrm{Cov}(I_n)
\end{equation*}
with $\mathrm{Cov}(I_n)_{ij} = \langle \pi(e_i),\pi(e_j)\rangle_{\frac{1}{x}}$, where $\pi: \R^n \rightarrow \X \Ker[S^T]$ is the $\langle ., . \rangle_{\frac{1}{x}}$-orthogonal projection to $\X \Ker[S^T]$.
\end{theorem*}
The space $\X \Ker[S^T]$ is actually the tangent space $T_x \Vss \subset T_xX$ and thus $\pi$ is the projection of a tangent vector to the tangent space $T_x \Vss$.
Moreover, the canonical basis vectors $e_i$ appearing in the linear algebraic calculations turn out to be the canonical basis vectors $\frac{\partial}{\partial x_i}$ of the tangent space $T_xX$.
An illustration of this geometry is shown in Fig. \ref{fig:intro} and more background is discussed in general in Remarks \ref{rmk:geometry} and \ref{rmk:geometry_2} and for the case of quasi-thermostatic CRN in Remark \ref{rmk:geom2}.
\begin{figure}
    \centering
    \includegraphics[scale=0.25]{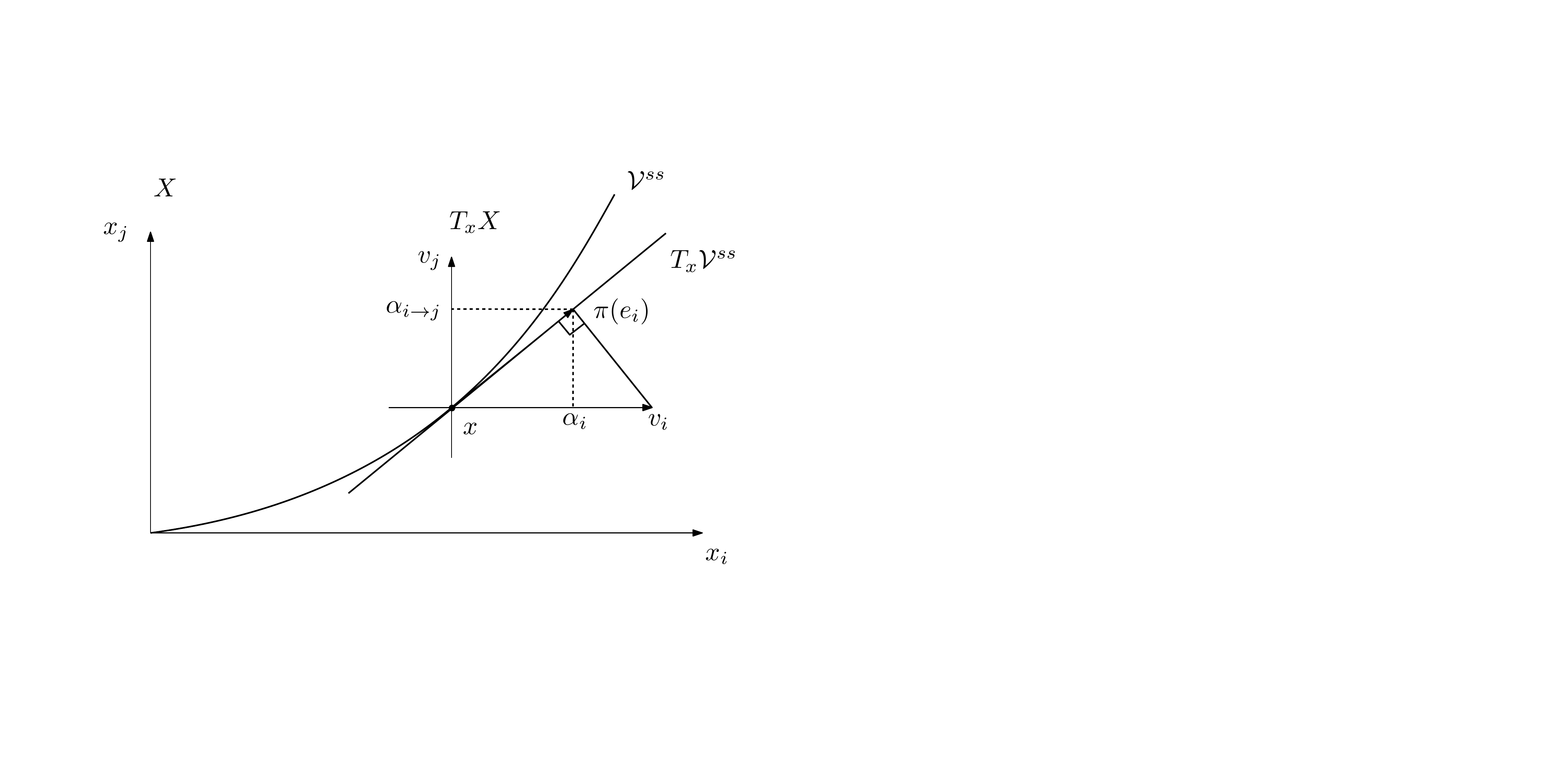}
    \caption{Illustration of the geometrical background for the case of quasi-thermostatic CRN.
    The concentration space $X$ is shown in two dimensions and the steady state manifold $\Vss$ is represented by a one-dimensional curve.
    The linear algebra takes place on the tangent spaces $T_x X$ and $T_x \Vss \subset T_x X$ at a given point $x \in X$.
    The vectors $e_i$ are the canonical basis vectors $\frac{\partial}{\partial x_i}$ of $T_xX$, and $\pi(e_i)$ is the $\langle ., . \rangle_{\frac{1}{x}}$-orthogonal projection of $e_i$ to $T_x \Vss$.
    The absolute sensitivities $\aij$ are the respective components of the vector $\pi(e_i)$.
    }
    \label{fig:intro}
\end{figure}

\section{Chemical reaction networks} \label{sec:CRN}

\subsection{General notions}

A chemical reaction network (CRN) is determined by the $n$ chemicals $X_1,\dotsc,X_n$ and $r$ reactions $R_1,\dotsc,R_r$ given by
\begin{align*}
    R_j: \sum_{i=1}^n S^-_{ij} X_i \rightarrow \sum_{i=1}^n S^+_{ij} X_i
\end{align*}
with nonnegative integer coefficients $S^+_{ij}$ and $S^-_{ij}$.
These stoichiometric coefficients determine the reactants and products of the reaction and the structure of the network is encoded in the $n \times r$ stoichiometric matrix $S = (S_{ij})$ with matrix elements
\begin{align*}
    S_{ij} = S^+_{ij}  - S^-_{ij}.
\end{align*}
The state of the CRN is given by the vector of positive concentration values $$x = (x_1,\dotsc,x_n)^T \in \mathbb{R}^n_{>0}, $$ where $x_i$ represents the concentration of the chemical $X_i$.
The state space is called concentration space and is denoted by $X := \mathbb{R}^n_{>0}$.
The dynamics of the CRN is governed by the equation
\begin{align*}
    \frac{\dd x}{\dd t} = Sj(x),
\end{align*}
where $j = (j_1,\dotsc,j_r)^T \in \mathbb{R}^r$ is the vector of reaction fluxes.
The specification of $j$ as a function of $x$ is tantamount to the choice of a kinetic model, with mass action kinetics being the most common one.
A state $x \in X$ which satisfies $\frac{\dd x}{\dd t} = Sj(x) = 0$ is called a steady state of the CRN.
The set
\begin{align*}
    \Vss = \{ x \in X | Sj(x) = 0 \}
\end{align*}
is assumed to be  a manifold (possibly with singularities). It is called the steady state manifold.

\subsection{Conserved quantities and sensitivity} \label{sec:sensitivity}
For any vector $\uu \in \Ker[S^T]$, the quantity $\langle \uu, x \rangle$ is conserved by the reaction dynamics, which follows from
\begin{align*}
    \frac{\dd \langle \uu, x \rangle}{\dd t} = \langle \uu, Sj(x) \rangle = \langle S^T \uu, j(x) \rangle = 0.
\end{align*}
Let $q$ denote the dimension of $\Ker[S^T]$, choose a basis $\{\uu_i\}_{i=1}^q$ of $\Ker[S^T]$, and write $\UU = (\uu_1,\dotsc,\uu_q)$ for the respective $n \times q$ matrix of basis vectors.
This yields the map
\begin{equation} \label{eq:eta_def}
    U^T : X \rightarrow \R^q.
\end{equation}
The vector $\eta := \UU^T x \in \R^q$ is conserved by the reaction dynamics and is called the vector of {\it conserved quantities} in CRN theory.
Conserved quantities result, for example, from the conservation of mass and larger molecular residues such as amino acids throughout all reactions of the CRN.
For any initial condition $x_0 \in X$ with $\eta := \UU^T x_0$, the reaction dynamics is confined to the {\it stoichiometric polytope}, which is defined as
\begin{align*}
    P(\eta) := \{ x \in \R^n_{\geq 0} | \UU^T x = \eta \}.
\end{align*}
The range of physically meaningful parameters $\eta$ is given by
\begin{align*}
\HH := \UU^T X \subset \R^q,
\end{align*}
which is an open submanifold of $\R^q$ of full dimension.
This gives a fibration of the concentration space $X$ by the stoichiometric polytopes $P(\eta)$ with the base space $H$.

If, locally at $x \in \Vss$, the map $U^T: \Vss \rightarrow H$ has a differentiable inverse, then the $n \times q$ {\it sensitivity matrix} $\chi$ with matrix elements
\begin{align} \label{eq:chi}
    \chi_{ij} = \frac{\partial x_i}{\partial \eta_j}
\end{align}
is well-defined.
It quantifies the infinitesimal change of the concentration values around $x$ with respect to infinitesimal changes in the values of the conserved quantities $\eta$.
The numerical values $\chi_{ij}$, however, depend on the choice of a basis for $\Ker[S^T]$ and are not absolute physical quantities.
It is the main purpose of this article to define sensitivities which are independent of choice of basis and study their properties.
The definition is given in the next section for the most general case which requires a locally differentiable parametrization of $\Vss$ by the conserved quantities but makes no further assumptions.
For quasi-thermostatic CRN more explicit results are available.
They are presented in Section \ref{sec:CRB}.

\section{Absolute sensitivity} \label{sec:abs_sens}

Absolute sensitivities are local quantities at a point $x \in \Vss \subset X$ which measure the sensitivity of the steady state concentrations $x_1,x_2,\dotsc,x_n$ with respect to infinitesimal concentration changes of the chemicals.
The definition requires that locally at $x$, the map $U^T : X \rightarrow H$ has a continuously differentiable inverse with image in $\Vss$.
If this is not the case, absolute sensitivities are not well-defined and one needs to restrict the parameter space $H$ accordingly.
This happens, for example, at bifurcation points.
Therefore, let $\tilde{H} \subset H$ be a submanifold of $H$ such that there is a differentiable section
\begin{equation*}
    \beta: \tilde{H} \rightarrow X
\end{equation*}
to $U^T: X \rightarrow H$ which satisfies $\Img[\beta] \subset \Vss$ and is a local inverse to $U^T: \Vss \rightarrow H$.
For example, if $x$ is a nondegenerate steady state, the existence of the section follows from the implicit function theorem.
Moreover, it is well-known to exist on all of $H$ for quasi-thermostatic steady states, cf. Section \ref{sec:diff_geom}.
This implies that $\beta(\eta)$ must be a point in the intersection $P(\eta) \cap \Vss$ and the requirement for $\beta$ to be a local inverse amounts to requiring that $\beta(\eta)$ is an isolated point in $P(\eta) \cap \Vss$.
This is illustrated in Fig. \ref{fig:sensitivity}.
The sensitivity matrix $\chi$ defined in \eqnref{eq:chi} is the Jacobian of this map
\begin{equation*}
    D_{\eta}\beta = \chi.
\end{equation*}
\begin{figure}[tb]
    \centering
    \includegraphics[scale=0.25]{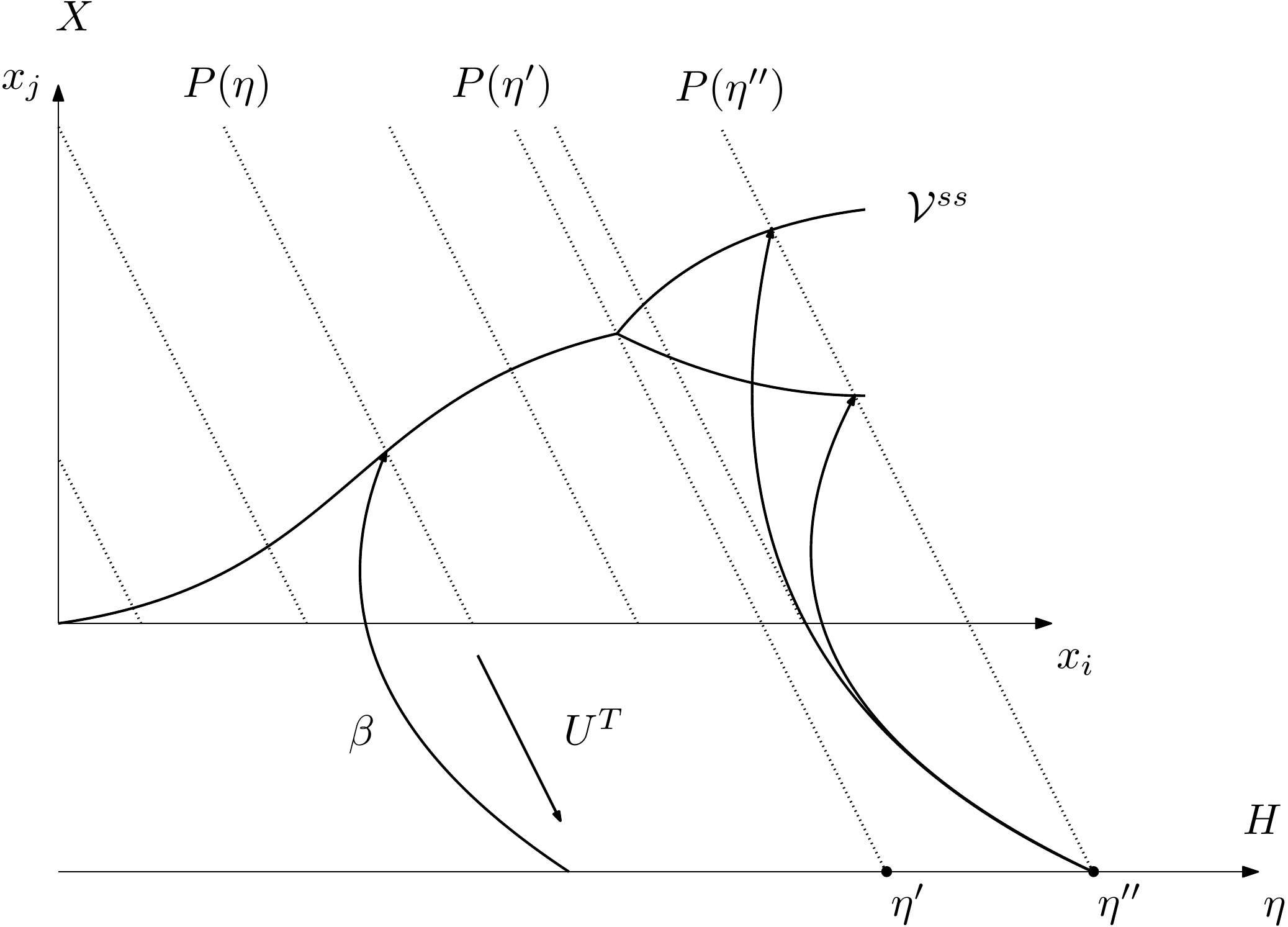}
    \caption{Setup for the definition of absolute sensitivities.
    The map $U^T: X \rightarrow H$ gives a fibration of $X$ by stoichiometric polytopes (indicated by dotted lines) with base $H$.
    The differentiable parametrization of $\Vss$ by the vectors of conserved quantities requires that the map $U^T: \Vss \rightarrow H$ is locally invertible with a differentiable inverse.
    The latter condition is not satisfied at, e.g., bifurcation points ($\eta'$ in the figure) and such points have to be excluded by restricting the parameter space $\tilde{H} := H \setminus \{ \eta' \}$ accordingly.
    If $\Vss$ has multiple intersection points with a stoichiometric polytope $P(\eta'')$, then there are different possible sections $\beta$ which amount to choosing one of the respective branches of $\Vss$.
    }
    \label{fig:sensitivity}
\end{figure}
With this setup, the absolute sensitivity of a chemical can be described as follows:
If, at a steady state $x \in \Vss$, the concentration of a single chemical $X_i$ is perturbed by an amount of $\delta x_i$, a new steady state is adopted by the CRN.
Thereby, the perturbation $\delta x_i$ distributes to concentration changes of all chemicals, prescribed by the coupling through the conserved quantities.
The absolute sensitivity $\alpha_i$ quantifies the fraction of concentration change that remains with $X_i$ after this redistribution, i.e. the concentration change of $X_i$ is given by $\alpha_i \delta x_i$, to first order in $\delta x_i$.
Analogously, the absolute cross-sensitivity $\aij$ quantifies the resulting changes of the concentration of the chemical $X_j$ as $\aij \delta x_i$, to first order in $\delta x_i$.

Let the steady state be given by $x = \beta(\eta)$.
The change $\delta x_i$ of the concentration of a chemical $X_i$ corresponds to the change in total concentration $\Delta x = (0, \dotsc, 0 ,\delta x_i, 0, \dotsc ,0)^T$.
This gives the change of the vector of conserved quantities $\Delta \eta = D_{x} U^T \Delta x = U^T \Delta x$.
One can express the adjusted steady state as $\beta(\eta + \Delta \eta)$ and thus obtain the linearization
\begin{align} \label{eq:Taylor}
    \beta(\eta + \Delta \eta) = \beta(\eta) + D_{\eta} \beta (\Delta \eta) + \mathcal{O}(\|\Delta \eta \|^2) = x + D_{\eta} \beta (U^T \Delta x) + \mathcal{O}(\|\Delta x \|^2).
\end{align}
The linear change in the concentration of $X_j$ is
\begin{align*}
    [D_{\eta} \beta (U^T \Delta x)]_j = \left[\frac{\partial x}{\partial \eta} U^T \Delta x \right]_j = \sum_{k=1}^q \frac{\partial x_j}{\partial \eta_k} u_{ik} \delta x_i,
\end{align*}
which leads to the following definition.
\begin{definition} \label{def:alpha_i}
The {\it absolute sensitivity $\aij$} of $X_j$ with respect to $X_i$ at a point $x \in \Vss$ is defined as
\begin{align*} \label{eq:aij}
    \aij:= \sum_{k=1}^q \frac{\partial x_j}{\partial \eta_k} u_{ik}
\end{align*}
and the {\it absolute sensitivity} of the chemical $X_i$ is $\alpha_i := \aii$.
The $n \times n$ matrix $\A$ of absolute sensitivities is given by
\begin{align*}
    \A_{ij} = \aji
\end{align*}
and the vector $\alpha$ of absolute sensitivities is given by the diagonal elements of $\A$, i.e., $\alpha = (\alpha_1,\dotsc,\alpha_n)^T \in \R^n$.
\end{definition}

\begin{remark} \label{rmk:A_chi_UT}
The matrix of absolute sensitivities is given by $\A = \chi U^T$.
\end{remark}
\begin{remark} \label{rmk:geometry}
More geometrically, the expansion in (\ref{eq:Taylor}) can be reformulated as follows:
It is natural to consider $\Delta x$ and $\Delta \eta$ as elements of the tangent spaces $T_xX$ and $T_{\eta}H$.
The tangent vector $\Delta \eta$ pulls back to the tangent vector $D_{\eta} \beta (\Delta \eta) \in T_x \Vss$.
The vector $D_{\eta} \beta (\Delta \eta)$ is the unique vector which induces the same change in the vector of conserved quantities as $\Delta x$ and which is, at the same time, tangent to $\Vss$ (the uniqueness of this tangent vector follows from the fact that $\beta$ is a section to $U^T$, i.e., $U^T \circ \beta = \textrm{id}_H$, therefore $U^T D_{\eta} \beta = I_q$ and thus the Jacobian $D_{\eta} \beta$ is injective).
Its $j$th component is given by $[D_{\eta} \beta (\Delta \eta)]_j = \aij \delta x_i$.
For $\delta x_i = 1$, one obtains $\Delta x = e_i$ and recovers the expression $\aij = [D_{\eta} \beta (U^T e_i)]_j = \left(\chi U^T\right)_{ji}$
\end{remark}

The absolute sensitivities have the following properties:

\begin{theorem} \label{thm:abs_sens1}
The matrix of absolute sensitivities $\A$ is independent of the choice of a basis of $\Ker[S^T]$.
Moreover, the equality
\begin{align*}
   \mathrm{Tr}[\A] = \sum_{i=1}^n \alpha_i = q
\end{align*}
holds, whereby $q = \dim \Ker[S^T]$.
\end{theorem}

\begin{proof}
The independence of on the choice of a basis of $\Ker[S^T]$ can be verified by a direct calculation:
Let $U'$ denote another matrix of basis vectors, i.e., $U' = UB$ for some  $B \in \textrm{GL}(q)$.
The respective vector of conserved quantities $\eta' = (U')^T x$ satisfies $\eta' = (U')^T x = B^T \eta$, where $\eta = U^T x$.
By Remark \ref{rmk:A_chi_UT}, the matrix of absolute sensitivities is given by
\begin{equation*}
    \A = \frac{\partial x}{\partial \eta}U^T = \frac{\partial x}{\partial \eta'} \frac{\partial \eta'}{\partial \eta} U^T = \frac{\partial x}{\partial \eta'} B^T U^T =  \frac{\partial x}{\partial \eta'} \left(U'\right)^T,
\end{equation*}
which proves the basis independence.

The second claim is verified by differentiating \eqnref{eq:eta_def}, i.e., $\eta_j = \sum_{i=1}^n u_{ij}x_i$, with respect to $\eta_j$ and summung over all $j$:
\begin{align} \label{eq:q}
    q = \sum_{i=1}^n \sum_{j=1}^q \frac{\partial x_i}{\partial \eta_j} u_{ij} = \sum_{i=1}^n \alpha_i.
\end{align}
\end{proof}
The equality $\mathrm{Tr}[\A] = q$ shows that for any CRN, there is a balance between low-sensitivity and high-sensitivity chemicals.

\begin{remark} \label{rmk:geometry_2}
The basis independence of $\A$ can also be seen from the geometry discussed in Remark \ref{rmk:geometry} without any calculations:
The tangent vector $D_{\eta} \beta (U^T e_i)$ is the unique vector tangent to $\Vss$ which satisfies $D_{\eta} \beta (U^T e_i) - e_i \in \Ker[U^T]$.
Now the decomposition $\R^n \cong \Ker[S^T] \oplus \Img[S] \cong \Img[U] \oplus \Ker[U^T]$ together with $\Img[U] \cong \Ker[S^T]$ implies that $\Ker[U^T] \cong \Img[S]$.
Thus the characterization of the tangent vector $D_{\eta} \beta (U^T e_i)$ can be written as $D_{\eta} \beta (U^T e_i) - e_i \in \Img[S]$ which shows that the geometrical construction is independent of the particular choice of $U$.
\end{remark}

This concludes the introduction of absolute sensitivity in the most general setup.
In the following section, a more explicit characterizations of the absolute sensitivity matrix is given when the steady state manifold is endowed with more structure.

\section{Generalized Cramer-Rao bound and absolute sensitivity for quasi-thermostatic CRN} \label{sec:CRB}

In this section, the absolute sensitivities are analyzed for quasi-thermostatic CRN \cite{horn1972necessary,feinberg1972complex}.
This class of CRN derives its importance from the fact that it includes all equilibrium and complex balanced CRN under mass action kinetics.
The class of quasi-thermostatic CRN is, however, even wider, cf. \cite{craciun2022disguised}.

\subsection{Quasi-thermostatic steady states}

Depending on the CRN and on the kinetic model, the shape of the steady state manifold $\Vss$ can be very complex and, in general, its global structure cannot be determined.
However, there is a large class of CRN whose steady state manifolds have the simple form
\begin{align} \label{eq:quasiTS_Vss}
	\Vss = \{ x \in X | \log x - \log x^{ss} \in \Ker[S^T] \},
\end{align}
where $x^{ss} \in X$ is a particular solution of $Sj(x) = 0$ (note that $\Vss$ is independent of the choice of the base point $x^{ss} \in \Vss$).
Such CRN are called {\it quasi-thermostatic} \cite{horn1972necessary,feinberg1972complex}.

Using the basis $\UU$ of $\Ker[S^T]$ defined in Section \ref{sec:sensitivity}, the steady state manifold of a quasi-thermostatic CRN can be parametrized by $\R^q$ as
\begin{align} \label{eq:quasiTS_Vss1}
    \gamma: \R^q &\rightarrow X \\
    \nonumber
    \lambda &\mapsto x^{ss} \circ \exp(\VV \lambda).
\end{align}
This parametrization is reminiscent of an exponential family of probability distributions\footnote{An exponential family of probability distributions has the form (\ref{eq:quasiTS_Vss1}), except for a normalization constant which ensures that the probabilities sum to 1.}, which is often encountered in information geometry \cite{amari2016} and the Cramer-Rao bound derived in Section \ref{sec:CRB_derivation} is based on this link between CRN theory and statistics.
The variable $\lambda$ is closely related to the chemical potential \cite{kobayashi2022kinetic}.

For quasi-thermostatic CRN, the intersection between the steady state manifold $\Vss$ and any given stoichiometric polytope $P(\eta)$ is unique, i.e., $\eta$ is a global coordinate for $\Vss$.
This is the content of Birch's theorem \cite{craciun2009toric} which has been proven in the context of CRN by Horn and Jackson in \cite{horn1972} and is well-known in algebraic statistics \cite{pachter2005}.
In other words, there is a parametrization of $\Vss$ by the space $\HH$ of conserved quantities given by
\begin{align} \label{eq:quasiTS_Vss2}
    \beta: \HH &\rightarrow X \\
    \nonumber
    \eta &\mapsto \Vss \cap P(\eta).
\end{align}
This parametrization plays an important role in applications because the $\eta$ variables are easy to control in experimental setups.
The Jacobian of $\beta$ is the sensitivity matrix $\chi$.
Although there is no analytical expression for $\beta(\eta)$ for general quasi-thermostatic CRN, the Jacobian can be explicitly computed, as is shown in the next section.

\subsection{Differential geometry of quasi-thermostatic CRN and sensitivity} \label{sec:diff_geom}
The two parametrizations of $\Vss$ introduced in the previous section are illustrated in Fig. \ref{fig:parametrizations}.
The steady state manifold can be either be explicitly parametrized by $\lambda \in \R^q$ or implicitly by giving the stoichiometric polytope which contains the point $x \in \Vss$.

\begin{figure}
    \centering
    \includegraphics[scale=0.25]{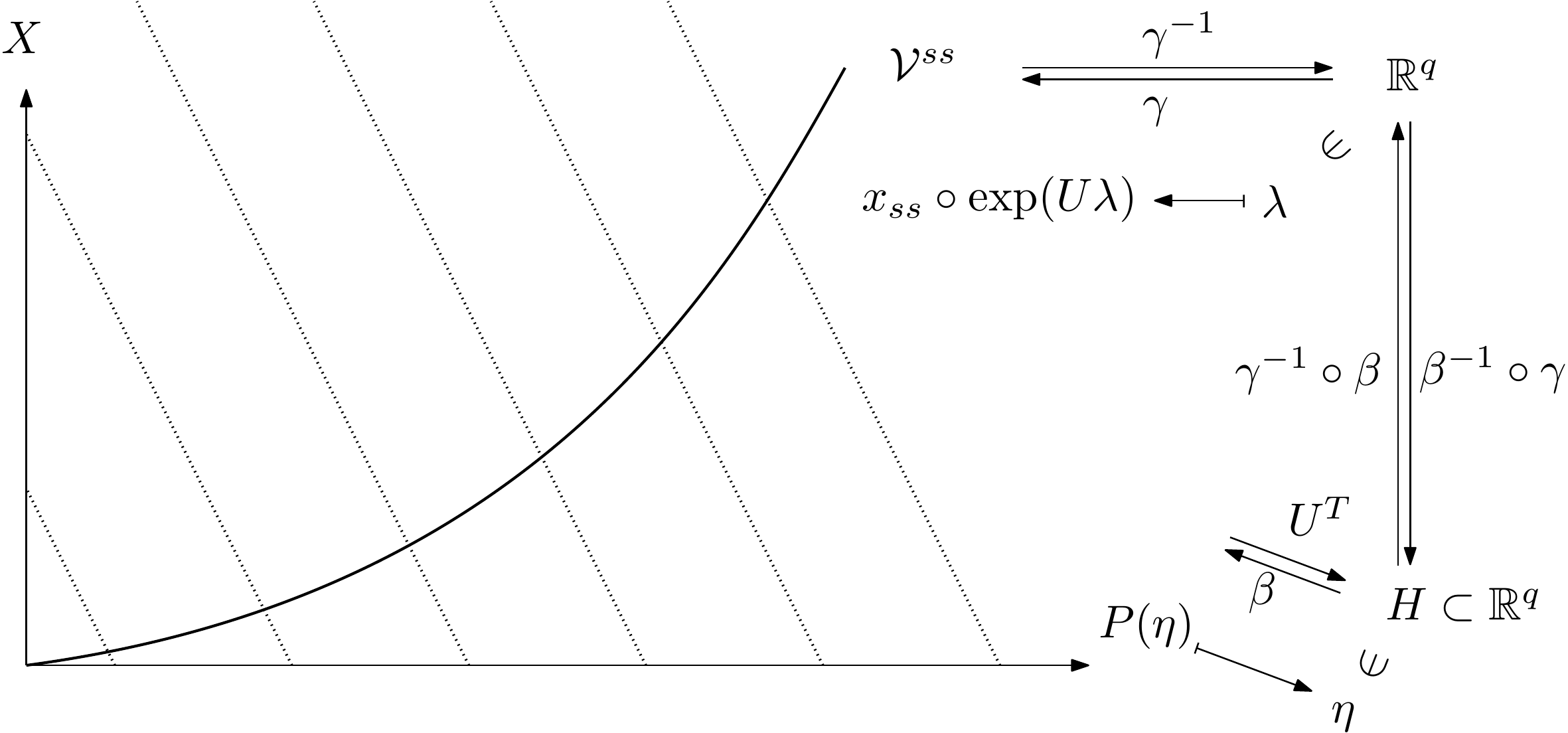}
    \caption{The two parametrization of $\Vss$ for quasi-thermostatic CRN.
    The map $\gamma$ provides a direct parametrization of $\Vss$ which is known as the exponential family in statistics and as a toric variety in algebraic geometry.
    The map $\beta$ characterizes points on $\Vss$ through their vector of conserved quantities $\eta$, i.e., as the unique intersection point between $\Vss$ and the stoichiometric polytope $P(\eta)$.
    }
    \label{fig:parametrizations}
\end{figure}

Fix a point $x \in \Vss$ with respective parameters $\lambda \in \R^q$ and $\eta \in H$.
The sensitivity matrix $\chi = \frac{\partial x}{\partial \eta}$ at $x$ is the Jacobian $D_{\eta}\beta$ of the map $\beta$ at $\eta = \UU^T x$.
This Jacobian can be evaluated by using the commutativity of the diagram in Fig. \ref{fig:parametrizations} as follows
\begin{align}
    \nonumber
    D_{\eta} \beta &= D_{\lambda} \gamma \cdot D_{\eta}(\gamma^{-1} \circ \beta) = D_{\lambda}\gamma \cdot [D_{\lambda}(\beta^{-1} \circ \gamma)]^{-1} \\
    \label{eq:Jacobian}
    &= D_{\lambda}\gamma \cdot [D_{x} \beta^{-1} \cdot D_{\lambda} \gamma]^{-1}.
\end{align}
The map $\beta^{-1}$ is the linear map given by the matrix $U^T$ and the Jacobian $D_{\lambda}\gamma$ can be evaluated explicitly from \eqnref{eq:quasiTS_Vss1} as $D_{\lambda}\gamma = \X \VV$, where $\X$ is the  $n \times n$ diagonal matrix with $\X_{ii} = x_i$.
This yields the explicit form for the sensitivity matrix $\chi(x) = \X \VV \cdot [\UU^T \X \VV]^{-1}$ and for the matrix of absolute sensitivities
\begin{align} \label{eq:Jacobian_final}
    \A = \chi U^T = \X \VV \cdot [\UU^T \X \VV]^{-1} U^T
\end{align}
In the next section, it is shown how the matrix of absolute sensitivities appears as a part of a Cramer-Rao bound for CRN and how the bound leads to an explicit linear algebraic characterization of the matrix elements.
The explicit expression for the matrix $A$ given in (\ref{eq:Jacobian_final}) yields $A^2 = A$.
This means that $A$ represents a projection operator and this aspect will also be further clarified based on the Cramer-Rao bound.

\subsection{The Cramer-Rao bound for quasi-thermostatic CRN} \label{sec:CRB_derivation}

The analogy between a concentration vector $x \in \R^n_{>0}$ as a distribution on the $n$-point set and a probability distribution is the core reason for the applicability of statistical and information geometric methods to CRN theory.
Thereby, the steady state manifold of a quasi-thermostatic CRN is analogous to the exponential family of probability distributions and thus the formulation of a Cramer-Rao bound for quasi-thermostatic CRN is natural.
In the bound derived here, the Jacobian of the coordinate change from $\eta$ to $\lambda$ coordinates serves as a Fisher information metric and the matrix $\UU$ is treated as the bias of an estimator.
With these ingredients, the Cramer-Rao bound for CRN is derived in Theorem \ref{thm:CRB}.
The lower bound is closely related to the matrix of absolute sensitivities and a tightening of the bound leads to a linear algebraic characterization of the absolute sensitivities in Lemma \ref{lem:abs_sens2} and Theorem \ref{thm:matrix_abs_sens}.

From now on, fix a point $x \in \Vss$ with coordinates $\lambda \in \R^q$ and $\eta \in H$, respectively.
Denote the Jacobian of the coordinate change from $\eta$ to $\lambda$ by
\begin{align} \label{eq:g_eta}
    g_{\eta}:= D_{\eta}(\gamma^{-1} \circ \beta) = [\UU^T \diag(x) \VV]^{-1}.
\end{align}
Moreover, define the $n \times n$ diagonal matrix $\XX$ by $\XX_{ii} = \frac{1}{x_i}$ and the $\XX$-weighted inner product on $\R^n$ by
\begin{align*}
    \langle v,w \rangle_{\frac{1}{x}} := \sum_{i=1}^n  \frac{1}{x_i} v_i w_i.
\end{align*}
Let $\G$ be an arbitrary $n \times n$ matrix and $\overline{\G}$ a $n \times n$ matrix whose column span satisfies
\begin{equation*}
    \textrm{Span}\left[\overline{\G}\right] \subset \textrm{Im}[S].
\end{equation*}
This is equivalent to saying that the columns of $\overline{\G}$ are orthogonal to $\X \Ker[S^T]$ with respect to the $\langle.,.\rangle_{\frac{1}{x}}$ inner product, i.e., $\overline{\G}^T\XX \X U = \overline{\G}^T U = 0$.
The covariance matrix of $\G$ is defined as
\begin{align*}
    \mathrm{Cov}(\G) := (\G - \overline{\G})^T \XX (\G - \overline{\G}).
\end{align*}
It is called a covariance matrix because its elements are of the form
\begin{align*}
    \mathrm{Cov}(\G)_{ij} := \langle \G_i - \overline{\G}_i,  \G_j - \overline{\G}_j \rangle_{\frac{1}{x}}.
\end{align*}
The following theorem is formally analogous to a Cramer-Rao bound for the covarience matrix $\mathrm{Cov}(\G)$.
\begin{theorem} \label{thm:CRB}
For a quasi-thermostatic CRN, let the covariance matrix $\mathrm{Cov}(\G)$ be defined as above.
It is bounded from below by
\begin{equation} \label{eq:CRB}
    \mathrm{Cov}(\G) \geq \G^T U g_{\eta} U^T \G,
\end{equation}
\noindent where the matrix inequality is understood in the sense that the difference matrix between the left hand side and the right hand side of the inequality is positive semidefinite.
\end{theorem}
\noindent

The proof is given in the Supplementary Material, Section \ref{SIapp:proof}.
It is an adaptation of the proof from \cite{kay1993} for the Cramer-Rao bound in exponential families.
This multivariate bound is, in a certain sense, orthogonal to the scalar bound derived for CRN dynamics in \cite{yoshimura2021}.

Using the expression (\ref{eq:Jacobian_final}) for the matrix of absolute sensitivities and substituting the expression (\ref{eq:g_eta}) for $g_{\eta}$, the Cramer-Rao bound can be rewritten as
\begin{equation*}
    \mathrm{Cov}(\G) \geq \G^T \XX \A \G
\end{equation*}
and if one chooses $\G$ to be the identity matrix $I_n$, the inequality
\begin{equation} \label{eq:CRBchi}
    \mathrm{Cov}(I_n) \geq \XX \A
\end{equation}
is obtained.

\subsection{Tightening the bound and a linear algebraic characterization of absolute sensitivities} \label{sec:tightening}

For $\overline{\G} = 0$, the diagonal elements of the inequality (\ref{eq:CRBchi}) yield $1 \geq \alpha_i$ which is not tight as can be seen by summing over all $i$ and comparing with \eqnref{eq:q}, giving $n \geq q$.
For a general $\G$ matrix, the optimal $\overline{\G}$ can be determined as follows.
The diagonal entries of $\textrm{Cov}(\G)$ are given by the squared norm $\| \G_i - \overline{\G}_i \|_{\frac{1}{x}}^2 = \langle \G_i - \overline{\G}_i,  \G_i - \overline{\G}_i \rangle_{\frac{1}{x}}$, which is minimized if and only if $\overline{\G}_i$ is the $\langle.,.\rangle_{\frac{1}{x}}$-orthogonal projection of $\G_i$ to $\textrm{Im}[S]$.
In this case $\G_i - \overline{\G}_i$ becomes the $\langle.,.\rangle_{\frac{1}{x}}$-orthogonal projection of $\G_i$ to $\X \Ker[S^T]$.
Denote this projection as
\begin{equation} \label{eq:pi}
    \pi: \R^n \rightarrow \X \Ker[S^T].
\end{equation}
The application of this to $\G = I_n$ yields the following linear algebraic characterization of the absolute sensitivities:
\begin{lemma} \label{lem:abs_sens2}
For quasi-thermostatic CRN, the absolute sensitivity $\alpha_i$ at a point $x =(x_1,\dotsc,x_n)$ is given by
\begin{align*}
    \alpha_i = x_i \lVert \pi(e_i) \rVert_{\frac{1}{x}}^2,
\end{align*}
where $e_i$ is the $i$th canonical unit vector.
\end{lemma}

\begin{proof}
   The Cramer-Rao bound (\ref{eq:CRBchi}) yields
   \begin{align} \label{eq:CRNinProof}
       x_i \lVert \pi(e_i) \rVert_{\frac{1}{x}} ^2 \geq \alpha_i.
   \end{align}
Expand $\pi(e_i)$ in an $\langle.,.\rangle_{\frac{1}{x}}$-orthonormal basis $\{v_1,...,v_q\}$ of $\X \Ker[S^T]$ as
\begin{align*}
    \pi(e_i) = \sum_{j=1}^q \langle e_i, v_j \rangle_{\frac{1}{x}} v_j = \frac{1}{x_i} \sum_{j=1}^q v_{ji} v_j.
\end{align*}
Then $\lVert \pi(e_i) \rVert_{\frac{1}{x}}^2 = \sum_{j=1}^q \frac{1}{x_i^2} v_{ji}^2$ and the inequality (\ref{eq:CRNinProof}) gives
\begin{equation*}
    \sum_{j=1}^q \frac{1}{x_i} v_{ji}^2 \geq \sum_{j=1}^q \frac{\partial x_i}{\partial \eta_j} u_{ij}
\end{equation*}
The summation over all $i$ yields
\begin{align*}
    \sum_{i=1}^n \sum_{j=1}^q \frac{1}{x_i} v_{ji}^2 \geq  \sum_{i=1}^n \sum_{j=1}^q \frac{\partial x_i}{\partial \eta_j} u_{ij}x_i.
\end{align*}
The left hand side is the sum of the squared norms $\sum_{j=1}^q \lVert  v_j \rVert_{\frac{1}{x}}^2 = q$ and the right hand side is equal to $q$ by equation (\ref{eq:q}).
Thus (\ref{eq:CRNinProof}) must be an equality, which proves the claim.
\end{proof}
This characterization of the $\alpha_i$ is independent of the choice of a basis for $\Ker[S^T]$.
It also yields the tightness of the Cramer-Rao bound for $\mathrm{Cov}(I_n)$ and thereby the following linear-algebraic characterization of the matrix of absolute sensitivities:
\begin{theorem} \label{thm:matrix_abs_sens}
For quasi-thermostatic CRN, the matrix of absolute sensitivities $\A$ at a point $x \in \Vss$ is given by
\begin{equation*} \label{eq:corr_abs_sens}
    \A = \X \mathrm{Cov}(I_n)
\end{equation*}
with $\mathrm{Cov}(I_n)_{ij} = \langle \pi(e_i),\pi(e_j)\rangle_{\frac{1}{x}}$.
Thus, the absolute sensitivities are given by
\begin{equation*}
    \aij = x_j \langle \pi(e_j),\pi(e_i)\rangle_{\frac{1}{x}}.
\end{equation*}
\end{theorem}
\begin{proof}
As in Lemma \ref{lem:abs_sens2}, choose $\overline{(I_n)}_i$ to be the $\langle.,.\rangle_{\frac{1}{x}}$-orthogonal projection of $e_i$ to $\Img[S]$.
Then, by the lemma, the matrix $\mathrm{Cov}(I_n) - \XX \A$ has only zero entries on the diagonal.
As it is positive semidefinite by Theorem \ref{thm:CRB}, it admits a Cholesky decomposition, which implies that the matrix is identically zero.
\end{proof}

\begin{remark} \label{rmk:geom2}
The explicit characterization of the absolute sensitivity $\aij = x_j \langle \pi(e_j), \pi(e_i) \rangle_{\frac{1}{x}}$ as stated in Theorem \ref{thm:matrix_abs_sens} can be recovered by geometrical arguments alone.
According to Remark \ref{rmk:geometry_2}, the vector $D_{\eta} \beta (U^T e_i)$ is the unique tangent vector of $\Vss$ which satisfies $D_{\eta} \beta (U^T e_i) - e_i \in \Img[S]$.
Now, the tangent space $T_x \Vss = \X \Ker[S^T]$ is $\langle.,.\rangle_{\frac{1}{x}}$-orthogonal to $\Img[S]$ and therefore $D_{\eta} \beta (U^T e_i)$ must be the $\langle.,.\rangle_{\frac{1}{x}}$-orthogonal projection of $e_i$ to $\X \Ker[S^T]$, i.e.,
\begin{equation} \label{eq:Delta_x_diff}
    D_{\eta} \beta (U^T e_i) = \pi(e_i).
\end{equation}
The $i$th component of a vector $v$ is given by $x_i \langle e_i, v \rangle_{\frac{1}{x}}$ and from Eq. (\ref{eq:Delta_x_diff}) one rederives the result from Theorem \ref{thm:matrix_abs_sens}
\begin{equation} \label{eq:abs_sens_rmk}
     \aij = [D_{\eta} \beta (U^T e_i)]_j = x_j \langle e_j, \pi( e_i) \rangle_{\frac{1}{x}} = x_j  \langle \pi(e_j), \pi( e_i) \rangle_{\frac{1}{x}},
\end{equation}
where the last equality follows from the fact that the linear form $\langle ., \pi( e_i) \rangle_{\frac{1}{x}}$ vanishes on $\Img[S]$ and thus $\langle e_j, \pi( e_i) \rangle_{\frac{1}{x}} = \langle \pi(e_j), \pi( e_i) \rangle_{\frac{1}{x}}$.\\

For a general steady state manifold $\Vss$, an analogue of the Theorem \ref{thm:matrix_abs_sens} can be formulated as follows:
Choose a symmetric and positive definite bilinear form $g$ on $T_xX$ such that $\Img[S]$ is orthogonal to $T_x\Vss$ with respect to the inner product $\langle u,v \rangle_g := u^Tgv$.
Then, $\aij$ is the $j$th component of $\pi (e_i)$, where $\pi : T_xX \rightarrow T_x \Vss$ is the $\langle .,. \rangle_g$-orthogonal projection.
This yields
\begin{equation} \label{eq:aij_general}
    \aij = \langle g^{-1}e_j, \pi (e_i) \rangle_g = \langle \pi (g^{-1}e_j), \pi (e_i) \rangle_g,
\end{equation}
where $g^{-1}$ is the inverse of the matrix representing $g$.
\end{remark}

Finally, for quasi-thermostatic CRN, the range for the absolute sensitivities $\ai$ is determined.
\begin{corollary}
    For quasi-thermostatic CRN, the absolute sensitivities $\ai$ satisfy
    \begin{equation*}
        \ai \in [0,1].
    \end{equation*}
\end{corollary}
\begin{proof}
    The bound arises $\ai \leq 1$ from comparing the diagonal elements in the inequality (\ref{eq:CRBchi}) and the bound $\ai \geq 0$ follows from the characterization $\alpha_i = x_i \lVert \pi(e_i) \rVert_{\frac{1}{x}}^2$ in Lemma \ref{lem:abs_sens2}.
\end{proof}

\begin{remark}
The absolute sensitivities for quasi-thermostatic CRN have the following symmetry property
\begin{equation*}
    x_i \aij = x_j \aji,
\end{equation*}
which is known as the condition of detailed balance for linear CRN with reaction rates $\aij$.
This condition characterizes equilibrium states and it would be interesting to investigate whether this analogy can be used to transfer results for equilibrium CRN to gain more insights into the properties of absolute sensitivities.
\end{remark}

\section{Example} \label{sec:example}

As an example, consider the core module of the IDHKP-IDH (IDH = isocitrate dehydrogenase, KP = kinase-phosphatase) glyoxylate bypass regulation system shown in the following reaction scheme:
\begin{equation} \label{eq:ACR}
\begin{tikzcd}
    \textrm{E} + \textrm{I}_{\textrm{p}}  \ar[r, rightharpoonup, shift left=.35ex,"k_1^+"] & \textrm{EI}_{\textrm{p}} \ar[l,shift left=.35ex, rightharpoonup,"k_1^-"] \ar[r, rightharpoonup, shift left=.35ex,"k_2^+"] & \textrm{E} + \textrm{I} \ar[l,shift left=.35ex, rightharpoonup,"k_2^-"] \\
    \textrm{EI}_{\textrm{p}} + \textrm{I} \ar[r, rightharpoonup, shift left=.35ex,"k_3^+"] & \textrm{EI}_{\textrm{p}}\textrm{I} \ar[l,shift left=.35ex, rightharpoonup,"k_3^-"] \ar[r, rightharpoonup, shift left=.35ex,"k_4^+"] & \textrm{EI}_{\textrm{p}} + \textrm{I}_{\textrm{p}}. \ar[l,shift left=.35ex, rightharpoonup,"k_4^-"]
\end{tikzcd}
\end{equation}
Here, I is the IDH enzyme, $\textrm{I}_{\textrm{p}}$ is its phosphorylated form, and E is the bifunctional enzyme IDH kinase-phosphatase.
The system has experimentally been shown to obey approximate concentration robustness in the IDH enzyme I \cite{laporte1985compensatory}.
In \cite{shinar2009robustness,shinar2010structural}, it was shown that if the CRN obeys mass action kinetics and the rate constants $k_2^-$ and $k_4^-$ are zero, i.e., if the respective fluxes vanish identically, then there is absolute concentration robustness in the IDH enzyme.

\subsection{Absolute sensitivity in the case of complex balancing}
However, the vanishing of fluxes is not consistent with thermodynamics and therefore, in this example, the concept of absolute sensitivity is used to analyze under which circumstances approximate concentration robustness persists in the chemical I if all reactions are reversible, i.e., the rate constants $k_2^-$ and $k_4^-$ are nonzero.
Abbreviate the chemicals as $X_1 = \textrm{E}, X_2 = \textrm{I}_{\textrm{p}}, X_3 = \textrm{EI}_{\textrm{p}}, X_4 = \textrm{I}, X_5 = \textrm{EI}_{\textrm{p}}\textrm{I}$ and use $x_i, i =1,\dotsc,5$ for the respective concentrations.
Under the assumption of complex-balancing\footnote{Note that the CRN has deficiency 1 and therefore the locus of rate constants for which the CRN is complex balanced is of codimension 1 in the space of all rate constants \cite{craciun2009toric}.
In this particular case, the condition on the rate constants is $k_1^+ k_2^+ k_3^+ k_4^+ = k_1^- k_2^- k_3^- k_4^-$.}, the absolute sensitivity for $X_4$ can be given in an analytically closed form based on the formula (\ref{eq:abs_sens_rmk}), i.e., $\alpha_4 = x_4 \langle e_4, \pi(e_4) \rangle$.
Explicitly, this yields
\begin{equation} \label{eq:a4}
    \alpha_4 = \frac{1}{1+r},
\end{equation}
where $r$ is given by the ratio
\begin{equation*}
    r = \frac{\left(x_2 + x_5 \right) \left(x_1 + x_3 \right) + x_1 \left( x_3 + 3x_5\right) +x_2x_5}{x_4 \left(x_1 + x_3 + x_5 \right)}.
\end{equation*}
A derivation of this expression is given in the Supplementary Material, Section \ref{app:details}.
Note that this expression is valid for any $x \in X$ whenever $x$ is a steady state point which satisfies complex balancing.
By adjusting the kinetic parameters, any point $x \in X$ can be made into such a point \cite{kobayashi2022kinetic}.

\subsection{Interpretation with respect to concentration robustness}
The functional form (\ref{eq:a4}) provides a rather straightforward understanding of the behavior  of the absolute sensitivity of $X_4$ because it is governed solely by the ratio $r$.
For $r \gg 0$, i.e., $$\left(x_2 + x_5 \right) \left(x_1 + x_3 \right) + x_1 \left( x_3 + 3x_5\right) +x_2x_5 \gg x_4 \left(x_1 + x_3 + x_5 \right),$$ the complex balanced CRN with reversible reactions is able to achieve very low sensitivities in $X_4$ and therefore mimic the behavior  of the irreversible CRN with absolute concentration robustness.
This is the case, for example, for $x_1 \approx x_2 \approx x_3 \approx x_5 \gg x_4$, for $x_2 \gg x_1 \approx x_4 \approx x_3 \approx x_5$ as well as for $x_1 \approx x_3 \gg x_2 \approx x_4 \approx x_5$, etc.
However, the CRN can also be very sensitive in $X_4$ whenever $r$ is close to $0$.
This happens, for example, when $x_4 \gg x_1 \approx x_2 \approx x_3 \approx x_5$.

\subsection{Numerical tests} \label{sec:numerical}
\paragraph{Sampling the whole concentration space}
To illustrate the capabilities of the absolute sensitivity concept, we have performed numerical tests with the CRN (\ref{eq:ACR}).
The absolute sensitivities $\alpha_i$ for all five chemicals of the CRN were computed at random points $x \in X$ in the range physiologically realistic concentrations, i.e., by uniform sampling of $\log x_i \in [-5,0]$.
The computation of the $\alpha_i$ was performed based on the formula (\ref{eq:Jacobian_final}):
\begin{equation*}
  \alpha_i = \left[\X \VV \cdot [\UU^T \X \VV]^{-1} U^T \right]_{ii}.
\end{equation*}
Figure \ref{fig:numerics} shows the resulting distributions of the $\alpha_i$ for 1.000.000 trials.
The mean values and the standard deviations are recorded in Table \ref{table:numerics}.
The results are rather surprising:
The IDH enzyme I and its phosphorylated form I$_{\textrm{p}}$ have the same values of absolute sensitivity: $\alpha_2, \alpha_4 \approx 0.295$.
Moreover, the distributions of $\alpha_2$ and $\alpha_4$ are very similarly shaped, with a clearly visible bias towards low values and with a decreasing number of occurrences towards higher values.
In contrast, the bifunctional enzyme E and its complex EI$_{\textrm{p}}$I with two IDH enzymes have rather high absolute sensitivity values: $\alpha_1, \alpha_5 \approx 0.514$.
Again, the distributions of $\alpha_1$ and $\alpha_5$ look almost identical.
This suggest a symmetry in the CRN with respect to sensitivity which is not immediately visible from the reaction scheme (\ref{eq:ACR}) - in particular E and EI$_{\textrm{p}}$I do not exhibit any obvious symmetry relation.
Notice that for a fixed concentration vector $x$, the relations $\alpha_2 = \alpha_4$ and $\alpha_1 = \alpha_5$ certainly do not hold true, cf. some sample results in Section \ref{SIsec:sample_alpha}.
Hence, the symmetry appears at the level of distributions.
These findings warrant further investigation in the future.

This example shows that absolute sensitivity can be easily computed numerically by using the linear algebraic characterization derived in this article, and how the concentration space can be sampled to get an intuition for the general sensitivity properties of the CRN.
For the core module of the IDHKP-IDH glyoxylate bypass regulation system it was indeed found that the IDH enzyme is rather robust in its concentration, which is in agreement with the known experimental facts.

To conclude this example, it is worth emphasizing that the absolute sensitivities are not only basis independent but that the explicit form (\ref{eq:a4}) allows to understand the absolute sensitivities of the chemicals on the whole concentration space without making a particular choice of kinetic rate constants.
However, notice that the rate constant enter implicitly as they determine state steady state manifold, and they have to be chosen such that the complex balancing condition is fulfilled.

\begin{figure}
    \centering
    \includegraphics[scale=0.4]{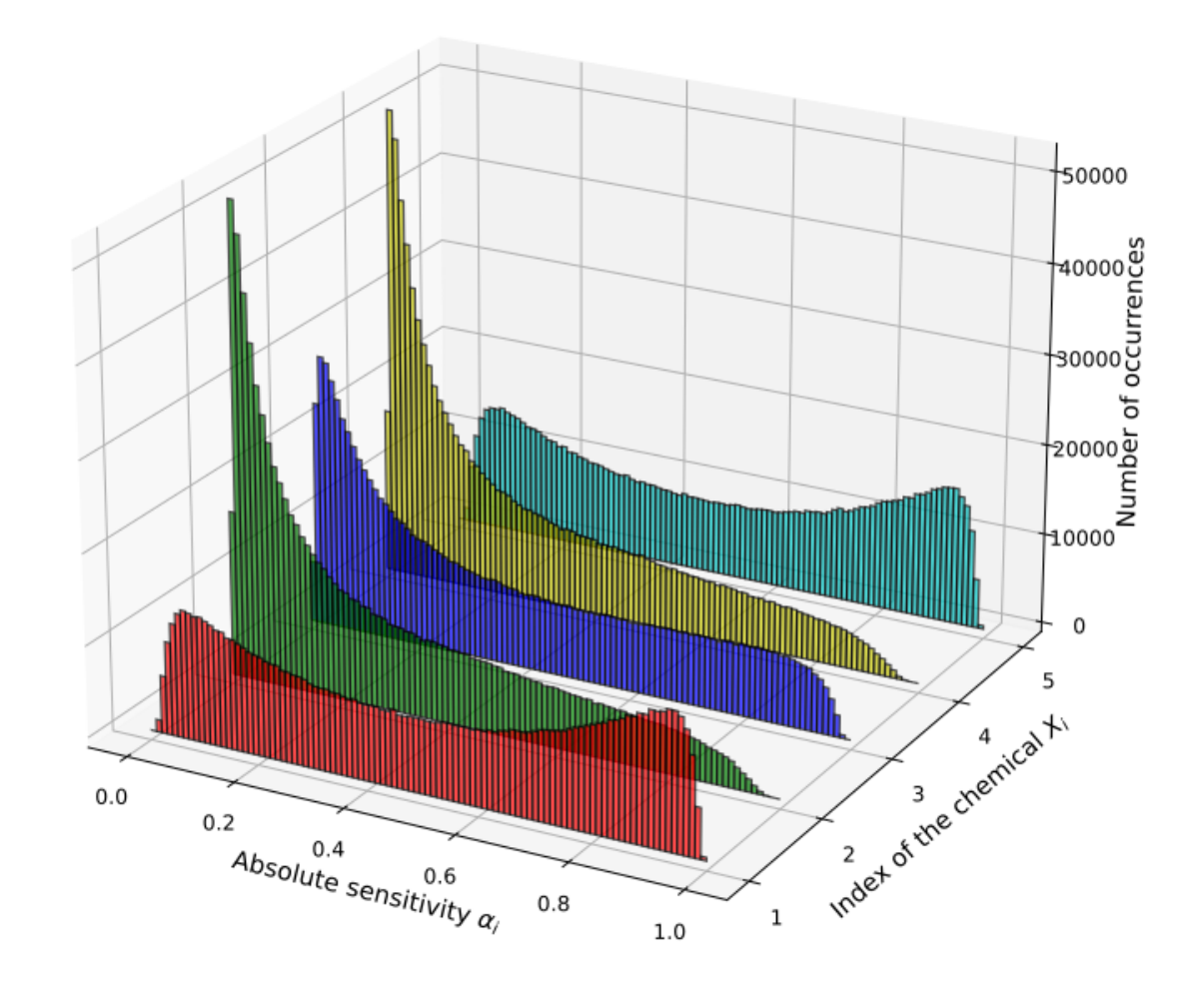}
    \caption{Distributions of the absolute sensitivities $\alpha_i$ for randomly sampled concentration vectors in $\log x^i \in [-5,0]$.
    Each histogram consists of 1.000.000 data points, distributed in 100 bins.
    }
    \label{fig:numerics}
\end{figure}
\begin{table}[htbp]
\footnotesize
  \caption{Mean values $\mu (\alpha_i)$ and standard deviations $\sigma (\alpha_i)$ for the distributions shown in Figure \ref{fig:numerics}.
    }\label{table:numerics}
    \begin{center}
    \begin{tabular}{|c||c|c|}
        \hline
        Chemical $X_i$ & $\mu (\alpha_i)$ & $\sigma (\alpha_i)$ \\
        \hline
        \hline
        $X_1$ (E) & 0.514 & 0.299 \\
        \hline
        $X_2$ (I$_{\textrm{p}}$) & 0.295 & 0.261 \\
        \hline
        $X_3$ (EI$_{\textrm{p}}$) & 0.383 & 0.291 \\
        \hline
        $X_4$ (I) & 0.295 & 0.261 \\
        \hline
        $X_5$ (EI$_{\textrm{p}}$I) & 0.514 & 0.300 \\
        \hline
    \end{tabular}
    \end{center}
\end{table}

\paragraph{Exploring fixed steady state manifolds}
In the previous paragraph, the whole concentration space was sampled whereby each random concentration vector was assumed to be a steady state vector.
This corresponds to varying steady state manifolds and thus varying the kinetic rate constants.
In the Supplementary Material, Section \ref{app:numerical}, it is demonstrated how to sample a fixed steady state manifold.
For any concrete choice of kinetic rate constants, the steady state manifold is two-dimensional and can, in general, be parametrized by the two conserved quantities $\eta_1:= x_1 + x_3 + x_5$ and $\eta_2 := x_2 + x_3 + x_4 + 2x_5$.
Some examples of how the absolute sensitivity $\alpha_4$ depends on $(\eta_1,\eta_2)$ for concrete choices of kinetic rate constants are given in the Supplementary Material, Section \ref{app:numerical}.
Notice that the parametrization $x^{ss}(\eta) = \beta(\eta)$ is accessible only numerically\footnote{It requires the solution of a polynomial of degree 5 in $x_3$ with constants in $\mathbb{Q}[k_i^{\pm},\eta_1,\eta_2]$.
In general, the situation is even more complex.}.\\

\section{Outlook} \label{sec:discussion}
The central technical contribution of this article is the derivation of the multivariate  Cramer-Rao bound, as presented in Section \ref{sec:CRB}.
Here, it has been derived by using linear algebra.
However, the natural setting for it is the recently developed information geometry for CRN \cite{yoshimura2021,kobayashi2022kinetic, sughiyama2022hessian,kobayashi2023information,loutchko2022}.
The information geometry of CRN is motivated by and can be seen as an analytic extension of the global algebro-geometric viewpoint of CRN theory \cite{craciun2008identifiability,craciun2009toric,gross2020joining,craciun2022disguised}.
There, the Cramer-Rao bound is obtained through the comparison of two Riemannian metrics.
One of the metrics is given by the Hessian of a strictly convex function on the space $X$ and the second one is its restriction to $\Vss \subset X$, hence the inequality.
Several properties that seem like coincidences in the setup of this article turn out to be valid in the Hessian geometric setup:
The metric $\XX$ generalizes to a metric $g$ given by the Hessian of a strictly convex function on the concentration space, the orthogonality between $T_x \Vss$ and $\Img[S]$ persists and the bound becomes strict for $\mathrm{Cov}(I_n)_{ij} = \langle \pi(e_j), \pi( e_i) \rangle_g$, where $\pi$ is the $g$-orthogonal projection to $T_x \Vss$ and $\langle.,.\rangle_g$ the respective bilinear product.

Moreover, this article introduces the notion of absolute sensitivity which remedies the basis dependence of the sensitivity matrix $\chi$.
Absolute sensitivities are defined purely through the geometry of concentration space and the embedding of the steady state manifold within it.
For quasi-thermostatic CRN, it has been shown that the absolute sensitivities $\ai$ of the chemicals $X_i$ lie in the interval $[0,1]$ and therefore they can be used to quantify approximate concentration robustness as well as highly sensitive chemicals.
Moreover, it shows how the strict requirements for absolute concentration robustness can be meaningfully relaxed:
Absolute concentration robustness in $X_i$ requires $\aji$ to vanish for all $j$, whereas the vanishing for some $j$ indicates insensitivity of $X_i$ to $X_j$ which is often a biochemically relevant situation.

Finally, going from quasi-thermostatic CRNs to more general CRNs allows $\ai > 1$ and $\aij <0$.
The condition $\ai > 1$ quantifies hypersensitivity in $X_i$.
The condition $\aij <0$ would allow to implement the operation of subtraction via CRN and thus enable various computations, logical circuits, and feedback regulations \cite{khammash2016engineering}.
Note that in \cite{feliu2019}, an example with $\aij <0$ is exhibited, by using the concept of sign-sensitivities and computational algebra techniques.
It would be especially exciting to obtain $\aij <0$ in reversible CRN because then it would be possible to investigate the thermodynamical cost of computation in CRN.
The geometrical characterization of $\aij$ as $\langle \pi (g^{-1}e_j), \pi (e_i) \rangle_g$ provided in Eq. (\ref{eq:aij_general}) is a helpful tool for finding such CRNs in the future.
It enables values of $\aij$ outside of the interval $[0,1]$ not through nonlinearities in kinetics but through non-ideal thermodynamic behavior which can already occur for equilibrium systems.
This generalization of the Cramer-Rao bound and absolute sensitivities to non-ideal systems is being explored in currently ongoing work.

An application of the concept of absolute sensitivity was presented in Section \ref{sec:numerical}:
It is easily possible to sample the absolute sensitivity values on the whole space of physiologically meaningful concentration vectors which span several orders of magnitude.
This is because the absolute sensitivities do not explicitly depend on on the kinetic rate constants but just on the particular point in concentration space, and because the formulae derived in this article are computationally very cheap.
This procedure yields distributions of absolute sensitivity and allows to gauge the robustness behavior of all the chemicals of the CRN.
For the example of the core module of the IDHKP-IDH glyoxylate bypass regulation system the experimentally known robustness in the IDH enzyme could be confirmed.
In addition, a hidden symmetry, which is not present at the level of individual point but at the level of distributions, was uncovered.
We are looking forward to apply the methodology to the analysis of other CRNs in the future.

\section*{Acknowledgments}

This research is supported by JST (JPMJCR2011, JPMJCR1927) and JSPS (19H05799).
Y. S. receives financial support from the Pub-lic\verb|\|Private R\&D Investment Strategic Expansion PrograM (PRISM) and programs for Bridging the gap between R\&D and the IDeal society (society 5.0) and Generating Economic and social value (BRIDGE) from Cabinet Office.
We thank Atsushi Kamimura, Shuhei Horiguchi and all other members of our lab for fruitful discussions.

\bibliographystyle{siamplain}
\bibliography{references}

\unappendix
\pagebreak
\begin{center}
\textbf{\large Supplementary Materials: Cramer-Rao bound and absolute sensitivity in chemical reaction networks}
\end{center}
\setcounter{equation}{0}
\setcounter{figure}{0}
\setcounter{table}{0}
\setcounter{section}{0}
\setcounter{page}{1}
\makeatletter
\renewcommand{\theequation}{S\arabic{equation}}
\renewcommand{\thefigure}{S\arabic{figure}}
\renewcommand{\thesection}{S\arabic{section}}

\noindent 
This Supplementary Material contains the proof of Theorem \ref{thm:CRB} from the main text in Section \ref{SIapp:proof}.
Sections \ref{app:details} and  \ref{app:numerical} provide additional details for the example from Section \ref{sec:example} in the main text.

\section{Proof of Theorem \ref{thm:CRB}} \label{SIapp:proof}

\begin{proof}
   This proof is an adaptation of the proof given in \cite{kay1993}, section 3B, for the Cramer-Rao bound in exponential families.
   Let $U = (u_1,...,u_q)$ be the $n \times q$ matrix of basis vectors of $\Ker[S^T]$ as in Section \ref{sec:sensitivity}.
   Write
   \begin{align*}
       \G^T U = (\G - \overline{\G})^T \XX^{1/2} \X^{1/2} U,
   \end{align*}
   where the equality follows from the orthogonality between $\overline{\G}$ and $U$.
   For any two vectors $a \in \mathbb{R}^n$ and $b \in \mathbb{R}^q$, the Cauchy-Schwarz inequality gives $$\left[ a^T (\G - \overline{\G})^T \XX (\G - \overline{\G}) a \right] \left[ b^T U^T \X U b \right] \geq \left[ a^T \G^T U b \right]^2 ,$$i.e.,
   \begin{equation*}
        \left[ a^T \mathrm{Cov}(\G) a \right] \left[ b^T g_{\eta}^{-1} b \right] \geq \left[ a^T \G^T U b \right]^2.
   \end{equation*}
   Choosing $$b = g_{\eta} U^T \G a$$ gives
   \begin{equation} \label{eq:proof}
   a^T \mathrm{Cov}(\G) a \left[ a^T \G^T U g_{\eta} U^T \G a \right] \geq \left[ a^T \G^T U g_{\eta} U^T \G a \right]^2.
   \end{equation}
   Writing $g_{\eta}^{-1} = [\X^{1/2}U]^T [\X^{1/2}U]$ shows that $g_{\eta}$ is positive semidefinite and thus the expression $$a^T \G^T U g_{\eta} U^T \G a $$ is non-negative.
   The theorem now follows from the inequality (\ref{eq:proof}).
\end{proof}

\section{Computation of $\alpha_4$ in Section \ref{sec:example}} \label{app:details}

The CRN is given in Scheme (\ref{eq:ACR}) by
\begin{equation} \label{SI:CRN1}
\begin{tikzcd}
    \textrm{E} + \textrm{I}_{\textrm{p}}  \ar[r, rightharpoonup, shift left=.35ex,"k_1^+"] & \textrm{EI}_{\textrm{p}} \ar[l,shift left=.35ex, rightharpoonup,"k_1^-"] \ar[r, rightharpoonup, shift left=.35ex,"k_2^+"] & \textrm{E} + \textrm{I} \ar[l,shift left=.35ex, rightharpoonup,"k_2^-"] \\
    \textrm{EI}_{\textrm{p}} + \textrm{I} \ar[r, rightharpoonup, shift left=.35ex,"k_3^+"] & \textrm{EI}_{\textrm{p}}\textrm{I} \ar[l,shift left=.35ex, rightharpoonup,"k_3^-"] \ar[r, rightharpoonup, shift left=.35ex,"k_4^+"] & \textrm{EI}_{\textrm{p}} + \textrm{I}_{\textrm{p}} \ar[l,shift left=.35ex, rightharpoonup,"k_4^-"]
\end{tikzcd}
\end{equation}
and with the abbreviations $X_1 = \textrm{E}, X_2 = \textrm{I}_{\textrm{p}}, X_3 = \textrm{EI}_{\textrm{p}}, X_4 = \textrm{I}, X_5 = \textrm{EI}_{\textrm{p}}\textrm{I}$ introduced in the main text it reads
\begin{equation} \label{SI:CRN2}
\begin{tikzcd}
    X_1 + X_2  \ar[r, rightharpoonup, shift left=.35ex,"k_1^+"] & X_3 \ar[l,shift left=.35ex, rightharpoonup,"k_1^-"] \ar[r, rightharpoonup, shift left=.35ex,"k_2^+"] & X_1 + X_4 \ar[l,shift left=.35ex, rightharpoonup,"k_2^-"] \\
    X_3 + X_4 \ar[r, rightharpoonup, shift left=.35ex,"k_3^+"] & X_5 \ar[l,shift left=.35ex, rightharpoonup,"k_3^-"] \ar[r, rightharpoonup, shift left=.35ex,"k_4^+"] & X_3 + X_2. \ar[l,shift left=.35ex, rightharpoonup,"k_4^-"]
\end{tikzcd}
\end{equation}
The variables $x_i, i =1,\dotsc,5$ are used for the respective concentrations of the chemicals $X_i$.\\

Now, the absolute sensitivity $\alpha_4$ is computed.
The stoichiometric matrix of the CRN is given by
\begin{equation*}
S=  \begin{bmatrix}
        -1 & 1 & 0 & 0 \\
        -1 & 0 & 0 & 1 \\
        1 & -1 & -1 & 1 \\
        0 & 1 & -1 & 0 \\
        0 & 0 & 1 & -1
    \end{bmatrix}
\end{equation*}
and the kernel of $S^T$ is spanned by $v_1 := (-1,1,0,1,1)$ and $v_2 := (1,0,1,0,1)$ and $\X \Ker[S^T]$ is spanned by $w_1 := \X v_1$ and $w_2 := \X v_2$.
To construct a $\langle.,.\rangle_{\frac{1}{x}}$ orthogonal basis $(w_2,w_3)$ of $\X \Ker[S^T]$, let $w_3 := w_2 + \lambda w_1$.
The requirement $\langle w_2, w_3\rangle_{\frac{1}{x}}$ yields
\begin{equation}
    \lambda = \frac{x_1 + x_3 + x_5}{x_1 - x_5}
\end{equation}
in the case that $x_1 \neq x_5$ (if $x_1 = x_5$, then $(w_1,w_2)$ is an $\langle.,.\rangle_{\frac{1}{x}}$-orthogonal basis already).
Now the projection $\pi(e_4)$ is given by
\begin{equation*}
    \pi(e_4) = \frac{\langle e_4, w_2 \rangle_{\frac{1}{x}}}{\langle w_2, w_2 \rangle_{\frac{1}{x}}} w_2 + \frac{\langle e_4, w_3 \rangle_{\frac{1}{x}}}{\langle w_3, w_3 \rangle_{\frac{1}{x}}} w_3 = \frac{\lambda}{\langle w_3, w_3 \rangle_{\frac{1}{x}}}w_3.
\end{equation*}
One computes $\langle w_3, w_3 \rangle_{\frac{1}{x}} = \lambda (x_5-x_1) + \lambda^2 (x_1 + x_2 + x_4 + x_5)$ and obtains, using the formula (\ref{eq:abs_sens_rmk}) from Remark \ref{rmk:geom2} for the absolute sensitivity\footnote{In the case $x_1 = x_5$, the analogous calculation can be done with $w_1$ instead of $w_3$ and the final results still holds true.}
\begin{align*}
    \alpha_4 &= x_4 \langle e_4, \pi(e_4) \rangle_{\frac{1}{x}} = \frac{x_4\lambda}{\langle w_3, w_3 \rangle_{\frac{1}{x}}} \langle e_4, w_3 \rangle_{\frac{1}{x}} = \frac{x_4\lambda^2}{\langle w_3, w_3 \rangle_{\frac{1}{x}}} \\
    &= \frac{x_4}{\lambda^{-1}(x_5-x_1) + (x_1 + x_2 + x_4 + x_5)} \\
    &=  \frac{x_4(x_1+x_3+x_5)}{(x_1+x_3+x_5) (x_1 + x_2 + x_4 + x_5)- (x_5-x_1)^2},
\end{align*}
which can be rearranged to yield the expression in the main text.

\section{Numerical examples} \label{app:numerical}

Here, the behavior of the absolute sensitivity $\alpha_4$ for the example in Section \ref{sec:example} is analyzed for different sets of kinetic rate constants.
In particular, it is known that absolute concentration robustness is achieved when the reactions $\textrm{EI}_{\textrm{p}} \rightarrow \textrm{E} + \textrm{I}$ and $\textrm{EI}_{\textrm{p}}\textrm{I} \rightarrow \textrm{EI}_{\textrm{p}} + \textrm{I}_{\textrm{p}}$ are irreversible \cite{shinar2010structural}.
Therefore, the focus of the numerical examples will be to investigate the effect of the smallness of the rate constants $k_2^-$ and $k_4^-$.

The assumption of complex balancing forces the CRN to be an equilibrium CRN as there are no cycles in the graph of complexes and reactions.
This implies the condition
\begin{equation}
    k_1^+ k_2^+ k_3^+ k_4^+ = k_1^- k_2^- k_3^- k_4^-
\end{equation}
on the rate constants.
In order to fix the rate constants for concrete examples, it is assumed that the order of all reactions is of a similar scale, i.e., $k_1^+ = k_2^+ = k_3^+ = k_4^+ = 1$. This leads to the condition $k_1^- k_2^- k_3^- k_4^- = 1$.
Moreover, it is assumed that the irreversibility is governed by
\begin{equation}
    \epsilon := k_2^- = k_4^-.
\end{equation}
For the condition close to irreversibility, $\epsilon = 0.01$ is chosen and contrasted with the case $\epsilon = 0.1$.
The complex balancing condition now reads $\epsilon^2 k_1^- k_3^- = 1$ and the three cases $k_1^- =  k_3^-$, $k_1^- = 100 k_3^-$, and $100 k_1^- =  k_3^-$ are considered.
The two conserved quantities
\begin{align}
\label{eq:eta1}
    \eta_1 &:= x_1 + x_3 + x_5 \\
\label{eq:eta2}
    \eta_2 &:= x_2 + x_3 + x_4 + 2x_5
\end{align}
quantify the total amount of $\textrm{E}$ and $\textrm{I}$, respectively and are natural parameters for the steady-state manifold.
The dependence of the absolute sensitivity $\alpha_4$ on $(\log \eta_1, \log \eta_2)$ is shown in Fig. \ref{fig:example}.
These results show that there is very little difference between $\epsilon = 0.01$ and $\epsilon = 1$.
However, the quotient $k_1^- / k_3^-$ has a strong influence on $\alpha_4$ and for $k_1^- \gg 100 k_3^-$ the lowest values for the absolute sensitivity are obtained.
These values are consistently low for all physiologically relevant values of $(\eta_1, \eta_2)$, where both conserved quantities range from $10^{-10}$ to $10$.

\begin{figure}
    \centering
    \includegraphics[scale=0.5]{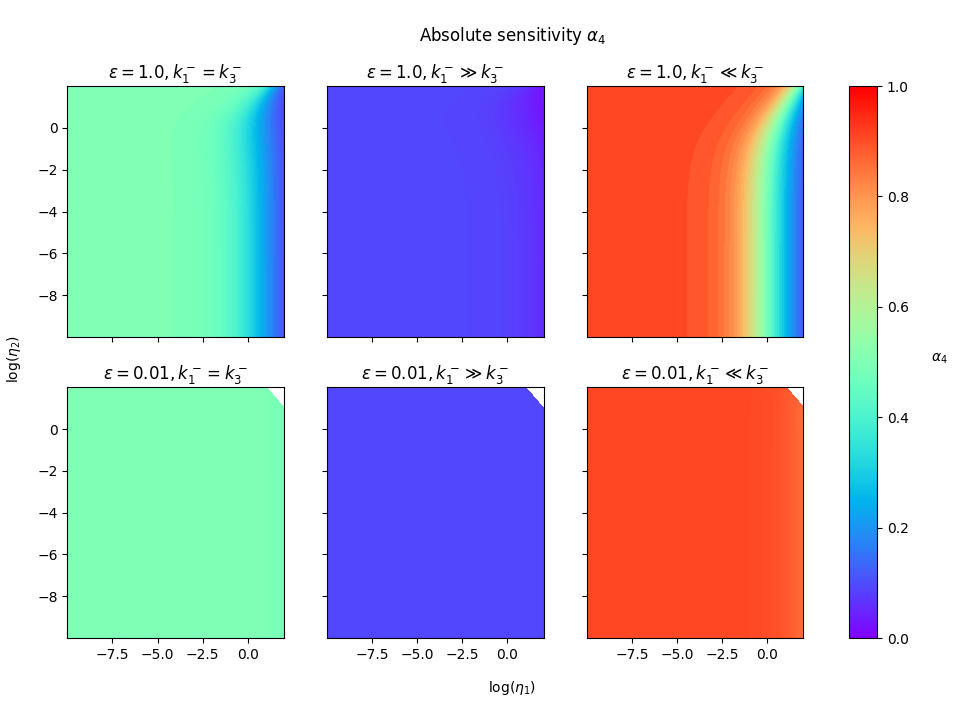}
    \caption{The dependence of the absolute sensitivity $\alpha_4$ on the kinetic rate constants.
    For this example, $k_1^+ = k_2^+ = k_3^+ = k_4^+ = 1$ are fixed, $\epsilon = k_2^- = k_4^-$ is varied as $\epsilon = 0.01$ and $\epsilon = 1$ and the three cases $k_1^- =  k_3^-$, $k_1^- = 100 k_3^-$, and $100 k_1^- =  k_3^-$ are considered under the condition $\epsilon^2 k_1^- k_3^- = 1$.
    The steady state manifold is two-dimensional and parametrized by $(\log \eta_1, \log \eta_2)$.
    }
    \label{fig:example}
\end{figure}

\subsection{Numerical procedure}

As the parametrization $x(\eta)$ is not available analytically, for each set of rate constants, a base point $x^{ss}$ is determined numerically by choosing arbitrary points $x \in X$ for $\eta_1 = \eta_2 =1$ and then by running mass action kinetics until convergence to $x^{ss}$.
The parametrization of the steady state manifold according to (\ref{eq:quasiTS_Vss1}), i.e.,
\begin{align}
    \gamma: \R^q &\rightarrow X \\
    \nonumber
    \lambda &\mapsto x^{ss} \circ \exp(\VV \lambda).
\end{align}
is used for $\lambda \in [10^{-15},10^3]^2$.
For each point $x = \gamma(\lambda)$, the absolute sensitivity $\alpha_4$ is computed according to the formula (\ref{eq:a4}) and the conserved quantities according to (\ref{eq:eta1}) and (\ref{eq:eta2}).
The results are shown in Fig. \ref{fig:example}.

\subsection{Sample vectors for Figure \ref{fig:numerics}} \label{SIsec:sample_alpha}
Each row in the following matrix corresponds to a tuple $(\alpha_1, \dotsc, \alpha_5)$ for a random concentration vector whose logarithms were sampled from a uniform distribution in the range $\log x_i \in [-5,0]$:
\begin{equation}
\begin{matrix}
0.571859 & 0.415994 & 0.499037 & 0.322429 & 0.190682\\
0.977996 & 0.052549 & 0.115131 & 0.166512 & 0.687813\\
0.837183 & 0.332138 & 0.374417 & 0.020179 & 0.436083\\
0.927398 & 0.069571 & 0.048338 & 0.048623 & 0.906070\\
0.466231 & 0.680558 & 0.532613 & 0.036203 & 0.284395\\
0.117028 & 0.314511 & 0.749724 & 0.279061 & 0.539675\\
0.433033 & 0.566505 & 0.510908 & 0.067377 & 0.422177\\
0.253371 & 0.618782 & 0.655395 & 0.215093 & 0.257359\\
0.942641 & 0.188410 & 0.069260 & 0.101602 & 0.698087\\
0.866876 & 0.026484 & 0.021009 & 0.736457 & 0.349174\\
0.036370 & 0.010242 & 0.913305 & 0.925176 & 0.114908\\
0.298352 & 0.067670 & 0.200336 & 0.669644 & 0.763998\\
0.568270 & 0.182549 & 0.247499 & 0.212768 & 0.788914\\
0.269326 & 0.668382 & 0.679059 & 0.036954 & 0.346279\\
0.182700 & 0.740687 & 0.574043 & 0.094458 & 0.408112\\
0.929341 & 0.005203 & 0.006780 & 0.156382 & 0.902294\\
0.693083 & 0.052616 & 0.625947 & 0.201384 & 0.426970\\
0.792792 & 0.074152 & 0.310112 & 0.494458 & 0.328485\\
0.672740 & 0.028854 & 0.363228 & 0.046930 & 0.888248\\
0.055413 & 0.199745 & 0.783556 & 0.629825 & 0.331460.\\
\end{matrix}.
\end{equation}
This shows that there is no symmetry of the form $\alpha_2 = \alpha_4$ and $\alpha_1 = \alpha_5$ for a fixed point $x$.
Hence, the symmetry seen in Figure \ref{fig:numerics} must appear at the level of distributions.

\end{document}